\documentclass[11pt]{article} 

\usepackage[margin=1in]{geometry} 
\usepackage{amsthm}
\newtheorem{lemma}{Lemma}
\newtheorem{theorem}{Theorem}
\newtheorem{definition}{Definition}
\usepackage{graphicx} 
\usepackage{amsmath}

\usepackage{booktabs} 
\usepackage{array} 
\usepackage{paralist} 
\usepackage{verbatim} 
\usepackage{subfig} 
\usepackage{amsthm}
\usepackage{amsmath}
\usepackage{amssymb}
\usepackage{fancyhdr} 
\usepackage{sectsty}
\usepackage{color}
\newcommand{\todo}[1]{}
\newcommand{\na}[1]{{\color{blue} Nick: #1}}

\allsectionsfont{\sffamily\mdseries\upshape} 
\usepackage{comment}
\usepackage[nottoc,notlof,notlot]{tocbibind} 
\usepackage[titles,subfigure]{tocloft} 

\newtheorem{example}{Example}

\usepackage{url}
\usepackage{setspace}
\usepackage{natbib}

\title{The (Non)-Existence of Stable Mechanisms in Incomplete Information Environments}
\author{Nick Arnosti\thanks{Dept. of Management Science \& Engineering, Stanford University. Work conducted at Microsoft Research.}, Nicole Immorlica \thanks{Microsoft Research.}, Brendan Lucier\footnotemark[2]}
\date{February 16, 2015}

\begin{document}


\maketitle

\begin{abstract}
We consider two-sided matching markets, and study the incentives of agents to circumvent a centralized clearing house by signing binding contracts with one another. It is well-known that if the clearing house implements a stable match and preferences are known, then no group of agents can profitably deviate in this manner. 

We ask whether this property holds even when agents have \emph{incomplete information} about their own preferences or the preferences of others. We find that it does not. In particular, when agents are uncertain about the preferences of others, \emph{every} mechanism is susceptible to deviations by groups of agents. When, in addition, agents are uncertain about their \emph{own} preferences, every mechanism is susceptible to deviations in which a single pair of agents agrees in advance to match to each other.

\end{abstract}

\section{Introduction} \vspace{-.1 in}
In entry-level labor markets, a large number of workers, having just completed their training, simultaneously seek jobs at firms.  These markets are especially prone to certain failures, including unraveling, in which workers receive job offers well before they finish their training, and exploding offers, in which job offers have incredibly short expiration dates.  In the medical intern market, for instance, prior to the introduction of the centralized clearing house (the {\it National Residency Matching Program}, or NRMP), medical students received offers for residency programs at US hospitals two years in advance of their employment date~\citep{roth-xing_1994}.  In the market for law clerks, law students have reported receiving exploding offers in which they were asked to accept or reject the position on the spot~\citep{roth-xing_1994}.

In many cases, including the medical intern market in  the United States and United Kingdom and the hiring of law students in Canada, governing agencies try to circumvent these market failures by introducing a centralized clearing house which solicits the preferences of all participants and uses these to recommend a matching \citep{roth_1991}. One main challenge of this approach is that of incentivizing participation.  Should a worker and firm suspect they each prefer the other to their assignment by the clearing house, then they would likely match with each other and not participate in the centralized mechanism.  Perhaps for this reason, clearing houses that fail to select a stable match have often had difficulty attracting participants and been discontinued \citep{roth_1991}. 

Empirically, however, even clearing houses which produce stable matches may fail to prevent early contracting. Examples include the market for Canadian law students \citep{roth-xing_1994} and the American gastroenterology match \citep{niederle-roth_2004, mckinney-niederle-roth_2005}. This is perhaps puzzling, as selecting a stable match ensures that no group of participants can profitably circumvent the clearing house ex-post.

Our work offers one possible explanation for this phenomenon. While stable clearing houses ensure that for \emph{fixed, known} preferences, no coalition can profitably deviate, in most natural settings, participants contemplating deviation do so without complete knowledge of others' preferences (and sometimes even their own preferences).  Our main finding is that in the presence of such uncertainty, \emph{no mechanism} can prevent agents from signing mutually beneficial side contracts.

We model uncertainty in preferences by assuming that agents have a common prior over the set of possible preference profiles, and may in addition know their own preferences. 
We consider two cases. In one, agents have no private information when contracting, and their decision of whether to sign a side contract depends only on the prior (and the mechanism used by the clearing house). In the second case, agents know their own preferences, but not those of others. When deciding whether to sign a side contract, agents consider their own preferences, along with the information revealed by the willingness (or unwillingness) of fellow agents to sign the proposed contract.

Note that with incomplete preference information, agents perceive the partner that they are assigned by a given mechanism to be a random variable. In order to study incentives for agents to deviate from the centralized clearing house, we must specify a way for agents to compare lotteries over match partners.
One seemingly natural model is that each agent gets, from each potential partner, a utility from being matched to that partner. When deciding between two uncertain outcomes, agents simply compare their corresponding expected utilities. Much of the previous literature has taken this approach, and indeed, it is straightforward to discover circumstances under which agents would rationally contract early (see Appendix \ref{appendix:insurance}). 
Such cases are perhaps unsurprising; after all, the central clearing houses that we study solicit only ordinal preference lists, while the competing mechanisms may be designed with agents' cardinal utilities in mind.

For this reason, we consider a purely ordinal notion of what it means for an agent to prefer one allocation to another. In our model, an agent debating between two uncertain outcomes chooses to sign a side contract only if the rank that they assign their partner under the proposed contract strictly first-order stochastically dominates the rank that they anticipate if all agents participate in the clearing house. This is a strong requirement, by which we mean that it is easy for a mechanism to be stable under this definition, relative to a definition relying on expected utility.  For instance, this definition rules out examples of beneficial deviations, such as that given in Appendix \ref{appendix:insurance}, where agents match to an acceptable, if sub-optimal, partner in order to avoid the possibility of a ``bad" outcome.  

Despite the strong requirements we impose on beneficial deviations, we show that every mechanism is vulnerable to side contracts when agents are initially uncertain about their preferences or the preferences of others.  On the other hand, when agents are certain about their own preferences but not about the preferences of others, then there do exist mechanisms that resist the formation of side contracts, when those contracts are limited to involving only a pair of agents (i.e., one from each side of the market).



 \vspace{-.1 in}

\section{Related Work} \label{sec:related} \vspace{-.1 in}
\cite{roth_1989} and \cite{roth-rothblum_1999} are among the first papers to model incomplete information in matching markets. These papers focus on the strategic implications of preference uncertainty, meaning that they study the question of whether agents should truthfully report to the clearinghouse. Our work, while it uses a similar preference model, assumes that the clearing house can observe agent preferences.  While this assumption may be realistic in some settings, we adopt it primarily in order to separate the strategic manipulation of matching mechanisms (as studied in the above papers) from the topic of early contracting that is the focus of this work.


Since the seminal work of \cite{roth-xing_1994}, the relationship between stability and unraveling has been studied using observational studies, laboratory experiments, and a range of theoretical models. Although some work concluded that stability played an important role in encouraging participation \citep{roth_1991, kagel-roth_2000}, other papers note that uncertainty may cause unraveling to occur even if a stable matching mechanism is used.

A common theme in these papers is that unraveling is driven by the motive of ``insurance." For example, the closely related models of  \cite{li-rosen_1998,suen_2000, li-suen_2000, li-suen_2004} study two-sided assignment models with transfers in which binding contracts may be signed in one of two periods (before or after revelation of pertinent information). In each of these papers, unraveling occurs (despite the stability of the second-round matching) because of agents' risk-aversion: when agents are risk-neutral, no early matches form. 

Even in models in which transfers are not possible (and so the notion of risk aversion has no obvious definition), the motive of insurance often drives early matching. The models presented by \cite{roth-xing_1994}, \cite{halaburda_2010}, and \cite{du-livne_2014} assume that agents have underlying cardinal utilities for each match, and compare lotteries over matchings by computing expected utilities. They demonstrate that unraveling may occur if, for example, workers are willing to accept an offer from their second-ranked firm (foregoing a chance to be matched to their top choice) in order to ensure that they do not match to a less-preferred option.\footnote{
In many-to-one settings, \cite{sonmez_1999} demonstrates that even in full-information environments, it may be possible for agents to profitably pre-arrange matches (a follow-up by \cite{afacan_2013} studies the welfare effects of such pre-arrangements). In order for all parties involved to strictly benefit, it must be the case that the firm hires (at least) one inferior worker in order to boost competition for their remaining spots (and thereby receive a worker who they would be otherwise unable to hire). Thus, the profitability of such an arrangement again relies on assumptions about the firm's underlying cardinal utility function.
}

While insurance may play a role in the early contracting observed by \cite{roth-xing_1994}, one contribution of our work is to show that it is not necessary to obtain such behavior. In this work, we show that even if agents are unwilling to forego top choices in order to avoid lower-ranked ones, they might rationally contract early with one another. Put another way, we demonstrate that some opportunities for early contracting may be identified on the basis of ordinal information alone (without making assumptions about agents' unobservable cardinal utilities).

The works of  \cite{manjunath_2013} and \cite{gudmundsson_2014} consider the stochastic dominance notion used in this paper; however they treat only the case (referred to in this paper as ``ex-post") where the preferences of agents are fixed, and the only randomness comes from the assignment mechanism. One contribution of our work is to define a stochastic dominance notion of stability under asymmetric information. This can be somewhat challenging, as agents' actions signal information about their type, which in turn might influence the actions of others.\footnote{The work of \cite{liu-mailath-postlewaite-samuelson_2014} has recently grappled with this inference procedure, and defined a notion of stable matching under uncertainty. Their model differs substantially from the one considered here: it takes a matching $\mu$ as given, and assumes that agents know the quality of their current match, but must make inferences about potential partners to whom they are not currently matched.}

Perhaps the paper that is closest in spirit to ours is that of \cite{peivandi-vohra_2013}, which considers the operation of a centralized exchange in a two-sided setting with transferrable utility. One of their main findings is that every trading mechanism can be blocked by an alternative; our results have a similar flavor, although they are established in a setting with non-transferrable utility.

 \vspace{-.1 in}

\section{Model and Notation} \vspace{-.1 in}

In this section, we introduce our notation, and define what it means for a matching to be \emph{ex-post}, \emph{interim}, or \emph{ex-ante} stable.

There is a (finite, non-empty) set $M$ of men and a (finite, non-empty) set $W$ of women. 

\begin{definition} \text{ }\\
Given $M$ and $W$, a {\bf matching} is a function $\mu : M \cup W \rightarrow M \cup W$ satisfying:
\begin{enumerate}
\item For each $m \in M$, $\mu(m) \in W \cup \{m\}$
\item For each $w \in W$, $\mu(w) \in M \cup \{w\}$
\item For each $m \in M$ and $w \in W$, $\mu(m) = w$ if and only if $\mu(w) = m$.
\end{enumerate}
\end{definition}

We let $\mathcal{M}(M,W)$ be the set of matchings on $M, W$.  \\

Given a set $S$, define $\mathcal{R}(S)$ to be the set of one-to-one functions mapping $S$ onto $\{1, 2, \ldots, |S|\}$. Given $m \in M$, let $P_m \in \mathcal{R}(W \cup \{m\})$ be $m$'s ordinal preference relation over women (and the option of remaining unmatched). Similarly, for $w \in W$, let $P_w \in \mathcal{R}(M \cup \{w\})$ be $w$'s ordinal preference relation over the men. We think of $P_m(w)$ as giving the \emph{rank} that $m$ assigns to $w$; that is, $P_m(w) = 1$ implies that matching to $w$ is $m$'s most-preferred outcome.

Given sets $M$ and $W$, we let $\mathcal{P}(M,W) =  \prod_{m \in M} \mathcal{R}(W \cup \{m\}) \times \prod_{w \in W} \mathcal{R}(M \cup \{w\})$ be the set of possible preference profiles. We use $P$ to denote an arbitrary element of $\mathcal{P}(M,W)$, and use $\psi$ to denote a probability distribution over $\mathcal{P}(M,W)$.  We use $P_A$ to refer to the preferences of agents in the set $A$ under profile $P$, and use $P_a$ (rather than the more cumbersome $P_{\{a\}}$) to refer to the preferences of agent $a$.

\begin{definition}\label{def:stable}
Given $M$ and $W$, and $P \in \mathcal{P}(M,W)$, we say that matching $\mu$ is {\bf stable at preference profile $P$} if and only if the following conditions hold.
\begin{enumerate}
\item For each $a \in M \cup W$, $P_a(\mu(a)) \leq P_a(a)$. 
\item For each $m \in M$ and $w \in W$ such that $P_m(\mu(m)) > P_m(w)$, we have $P_w(\mu(w)) < P_w(m)$.
\end{enumerate}
\end{definition}

This is the standard notion of stability; the first condition states that agents may only be matched to partners whom they prefer to going unmatched, and the second states that whenever $m$ prefers $w$ to his partner under $\mu$, it must be that $w$ prefers her partner under $\mu$ to $m$.

In what follows, we fix $M$ and $W$, and omit the dependence of $\mathcal{M}$ and $\mathcal{P}$ on the sets $M$ and $W$. We define a \emph{mechanism} to be a (possibly random) mapping $\phi : \mathcal{P} \rightarrow \mathcal{M}$. We use $A'$ to denote a subset of $M \cup W$.

We now define what it means for a coalition of agents to \emph{block} the mechanism $\phi$, and what it means for a \emph{mechanism} (rather than a matching) to be stable. Because we wish to consider randomized mechanisms, we must have a way for agents to compare lotteries over outcomes. As mentioned in the introduction, our notion of blocking relates to stochastic dominance. Given random variables $X, Y \in \mathbb{N}$, say that $X$ \emph{first-order stochastically dominates} $Y$ (denoted $X \succ Y$) if for all $n \in \mathbb{N}$, $\Pr(X \leq n) \geq \Pr(Y \leq n)$, with strict inequality for at least one value of $n$.

An astute reader will note that this definition reverses the usual inequalities; that is, $X \succ Y$ implies that $X$ is ``smaller" than $Y$. We adopt this convention because below, $X$ and $Y$ will represent the ranks assigned by each agent to their partner (where the most preferred option has a rank of one), and thus by our convention, $X \succ Y$ means that $X$ is preferred to $Y$.

\begin{definition}[Ex-Post Stability]  \label{def:expost}
Given $M, W$ and a profile $P \in \mathcal{P}(M,W)$, coalition $A'$ {\bf blocks mechanism $\phi$ ex-post} at $P$ if there exists a  mechanism $\phi'$ such that for each $a \in A'$,  \vspace{-.05 in}
\begin{enumerate}
\item $\Pr(\phi'(P)(a) \in A' )=1$, and \vspace{-.05 in}
\item $P_a(\phi'(P)(a)) \succ P_a(\phi(P)(a))$.  \vspace{-.05 in}
\end{enumerate}
Mechanism $\phi$ is {\bf ex-post stable at profile $P$} if no coalition of agents blocks $\phi$ ex-post at $P$. \\
Mechanism $\phi$ is {\bf ex-post stable} if it is ex-post stable at $P$ for all $P \in \mathcal{P}(M,W)$. \\
Mechanism $\phi$ is {\bf ex-post pairwise stable} if for all $P$, no coalition consisting of at most one man and at most one woman blocks $\phi$ ex post at $P$.
\end{definition}

Note that in the above setting, because $P$ is fixed, the mechanism $\phi'$ is really just a random matching.  The first condition in the definition requires that the deviating agents can implement this alternative (random) matching without the cooperation of the other agents; the second condition requires that for each agent, the random variable denoting the rank of his partner under the alternative $\phi'$ stochastically dominates the rank of his partner under the original mechanism.  


Note that if the mechanism $\phi$ is deterministic, then it is ex-post pairwise stable if and only if the matching it produces is stable in the sense of Definition~\ref{def:stable}.


The above notions of blocking and stability are concerned only with cases where the preference profile $P$ is fixed. In this paper, we assume that at the time of choosing between mechanisms $\phi$ and $\phi'$, agents have incomplete information about the profile $P$ that will eventually be realized (and used to implement a matching). 
We model this incomplete information by assuming that it is common knowledge that $P$ is drawn from a prior $\psi$ over $\mathcal{P}$.  Given a mechanism $\phi$, each agent may use $\psi$ to determine the ex-ante distribution of the rank of the partner that they will be assigned by $\phi$. This allows us to define what it means for a coalition to block $\phi$ ex-ante, and for a mechanism $\phi$ to be ex-ante stable.

\todo{Tweak notation? Or at least comment that randomness here is over both $\psi$, $\phi$ ($P$ and $\mu$), whereas with ex-post stability it is over only $\phi$.}
\begin{definition}[Ex-Ante Stability] \label{def:exante}
Given $M, W$ and a prior $\psi$ over $\mathcal{P}(M,W)$, coalition $A'$ {\bf blocks mechanism $\phi$ ex-ante} at $\psi$ if there exists a mechanism $\phi'$ such that  if $P$ is drawn from the prior $\psi$, then for each $a \in A'$, \vspace{-.05 in} 
\begin{enumerate}
\item $\Pr(\phi'(P)(a) \in A' )=1$, and \vspace{-.05 in} 
\item $P_a(\phi'(P)(a)) \succ P_a(\phi(P)(a))$. \vspace{-.05 in} 
\end{enumerate}
Mechanism $\phi$ is {\bf ex-ante stable at prior $\psi$} if no coalition of agents blocks $\phi$ ex-ante at $\psi$. \\
Mechanism $\phi$ is {\bf ex-ante stable} if it is ex-ante stable at $\psi$ for all priors $\psi$. \\
Mechanism $\phi$ is {\bf ex-ante pairwise stable} if, for all priors $\psi$, no coalition consisting of at most one man and at most one woman blocks $\phi$ ex-ante at $\psi$.
\end{definition}

Note that the only difference between ex-ante and ex-post stability is that the randomness in Definition \ref{def:exante} is over both the realized profile $P$ and the matching produced by $\phi$, whereas in Definition \ref{def:expost}, the profile $P$ is deterministic. Put another way, the mechanism $\phi$ is ex-post stable if and only if it is ex-ante stable at all deterministic distributions $\psi$.

The notions of ex-ante and ex-post stability defined above are fairly straightforward because the information available to each agent is identical. In order to study the case where each agent knows his or her own preferences but not the preferences of others, we must define an appropriate notion of a blocking coalition. In particular, if man $m$ decides to enter into a contract with woman $w$, $m$ knows not only his own preferences, but also learns about those of $w$ from the fact that she is willing to sign the contract. Our definition of what it means for a coalition to block $\phi$ in the interim takes this into account.

In words, given the common prior $\psi$, we say that a coalition $A' $ \emph{blocks $\phi$ in the interim} if there exists a preference profile $P$ that occurs with positive probability under $\psi$ such that when preferences are $P$, all members of $A'$ agree that the outcome of $\phi'$ stochastically dominates that of $\phi$, given their own preferences and the fact that other members of $A'$ also prefer $\phi'$. We formally define this concept below, where we use the notation $\psi( \cdot )$ to represent the probability measure assigned by the distribution $\psi$ to the argument.

\begin{definition}[Interim Stability]
\label{def:interim}
Given $M, W$, and a prior $\psi$ over $\mathcal{P}(M,W)$, coalition $A'$ {\bf blocks mechanism $\phi$ in the interim} if there exists a mechanism $\phi'$, and for each $a \in A'$, a subset of preferences $\mathcal{R}_a$ satisfying the following:
\begin{enumerate}
\item For each $P \in \mathcal{P}$,  $\Pr(\phi'(P)(a) \in A') = 1$.   \vspace{-.05 in} 
\item For each agent $a \in A'$ and each preference profile $\tilde{P}_a$, $\tilde{P}_a \in \mathcal{R}_a$ if and only if \vspace{-.05 in}
\begin{enumerate}
\item $\psi(Y_a(\tilde{P}_a)) > 0$, where $Y_a(\tilde{P}_a) = \{P \colon P_a = \tilde{P}_a\} \cap \{ P \colon P_{a'} \in \mathcal{R}_{a'}\  \forall a' \in A'  \backslash\{a\}\}$
\item When $P$ is drawn from the conditional distribution of $\psi$ given $Y_a(\tilde{P}_a)$, we have
$P_a(\phi'(P)(a)) \succ P_a(\phi(P)(a)) $.
\end{enumerate}
\end{enumerate}
Mechanism $\phi$ is {\bf interim stable at $\psi$} if no coalition of agents blocks $\phi$ in the interim at $\psi$.\\
Mechanism $\phi$ is {\bf interim stable} if it is interim stable at $\psi$ for all distributions $\psi$. \\
Mechanism $\phi$ is {\bf interim pairwise stable} if, for all priors $\psi$, no coalition consisting of at most one man and at most one woman blocks $\phi$ in the interim at $\psi$.
\end{definition}
To motivate the above definition of an interim blocking coalition, consider a game in which a moderator approaches a subset $A'$ of agents, and asks each whether they would prefer to be matched according to the mechanism $\phi$ (proposed by the central clearing house) or the alternative $\phi'$ (which matches agents in $A'$ to each other). Only if all agents agree that they would prefer $\phi'$ is this mechanism used.  Condition 1 simply states that the mechanism $\phi'$ generates matchings among the (potentially) deviating coalition $A'$.

We think of $\mathcal{R}_a$ as being a set of preferences for which agent $a$ agrees to use mechanism $\phi'$. The set $Y_a(\tilde{P}_a)$ is the set of profiles which agent $a$ considers possible, conditioned on the events $P_a = \tilde{P}_a$ and the fact that all other agents in $A'$ agree to use mechanism $\phi'$.
Condition 2 is a consistency condition on the preference subsets $\mathcal{R}_a$: 
 2a) states that agents in $A'$ should agree to $\phi'$ only if they believe that there is a chance that the other agents in $A'$ will also agree to $\phi'$ (that is, if $\psi$ assigns positive mass to $Y_a$); moreover, 2b) states that in the cases when $P_a \in \mathcal{R}_a$ \emph{and the other agents select} $\phi'$, it should be the case that $a$ ``prefers" the mechanism $\phi'$ to $\phi$ (here and in the remainder of the paper, when we write that agent $a$ prefers $\phi'$ to $\phi$, we mean that \emph{given the information available to $a$}, the rank of $a$'s partner under $\phi'$ stochastically dominates the rank of $a$'s partner under $\phi$).

We now move on to our main results.



 \vspace{-.1 in}

\section{Results} \vspace{-.1 in}

We begin with the following observation, which states that the three notions of stability discussed above are comparable, in that ex-ante stability is a stronger requirement than interim stability, which is in turn a stronger requirement than ex-post stability.

\begin{lemma} \label{lem:nesting} \text{ } \\
If $\phi$ is ex-ante (pairwise) stable, then it is interim (pairwise) stable. \\ If $\phi$ is interim (pairwise) stable, then it is ex-post (pairwise) stable.
\end{lemma}

\begin{proof}
We argue the contrapositive in both cases. Suppose that $\phi$ is not ex-post stable. This implies that there exists a preference profile $P$, a coalition $A'$, and a mechanism $\phi'$ that only matches agents in $A'$ to each other, such that all agents in $A'$ prefer $\phi'$ to $\phi$, given $P$. If we take $\psi$ to place all of its mass on profile $P$, then (trivially) $A'$ also blocks $\phi$ in the interim, proving that $\phi$ is not interim stable. 

Suppose now that $\phi$ is not interim stable. This implies that there exists a distribution $\psi$ over $\mathcal{P}$, a coalition $A'$, a mechanism $\phi'$ that only matches agents in $A'$ to each other, and preference orderings $\mathcal{R}_a$ satisfying the following conditions: the set of profiles $Y=\{P:\forall a\in A', P_a\in\mathcal{R}_a\}$ has positive mass $\psi \left(Y\right)>0$; and conditioned on the profile being in $Y$, agents in $A'$ want to switch to $\phi'$ , i.e., for all $a\in A'$ and {\it for all} $P_a \in \mathcal{R}_a$ agent $a$ prefers $\phi'$ to $\phi$ conditioned on the profile being in $Y$.  Thus, agent $a$ must prefer $\phi'$ even ex ante (conditioned only on $P \in Y$).

If we take $\psi'$ to be the conditional distribution of $\psi$ given $P \in Y$, it follows that under $\psi'$, all agents $a \in A'$ prefer mechanism $\phi'$ to mechanism $\phi$ ex-ante, so $\phi$ is not ex-ante stable.
\end{proof}

\subsection{Ex-post Stability}
We now consider each of our three notions of stability in turn, beginning with ex-post stability.  By Lemma \ref{lem:nesting}, ex-post stability is the easiest of the three conditions to satisfy.  Indeed, we show there not only exist ex-post stable mechanisms, but that any mechanism that commits to always returning a stable matching is ex-post stable.

\begin{theorem} \text{ } \label{thm:expost} \\
Any mechanism that produces a stable matching with certainty is ex-post stable. 
\end{theorem}
Note that if the mechanism $\phi$ is deterministic, then (trivially) it is ex-post stable if and only if it always produces a stable matching. Thus, for deterministic mechanisms, our notion of ex-post stability coincides with the ``standard" definition of a stable mechanism.  Theorem \ref{thm:expost} states further that any mechanism that randomizes among stable matchings is also ex-post stable. This fact appears as Proposition 3 in \citep{manjunath_2013}.\footnote{We thank an anonymous reviewer for the reference.}

We next show in Example \ref{ex:notnecessary} that the converse of Theorem \ref{thm:expost} does not hold.  That is, there exist randomized mechanisms $\phi$ which sometimes select unstable matches but are nevertheless ex-post stable. In this and other examples, we use the notation $P_m : w_1, w_2, w_3$ as shorthand indicating that $m$ ranks $w_1$ first, $w_2$ second, $w_3$ third, and considers going unmatched to be the least desirable outcome.


\begin{example} \label{ex:notnecessary}\[\begin{array}{l l l l}
P_{m_1} &: w_1 , w_2 , w_3 \hspace{.6 in} P_{w_1} &: m_3 , m_2 , m_1  \\
P_{m_2} &:  w_1 , w_3 , w_2 \hspace{.6 in} P_{w_2} &: m_2 , m_1 , m_3 \\
P_{m_3} &: w_2 , w_1 , w_3 \hspace{.6 in} P_{w_3} &: m_3 , m_2 , m_1
\end{array}\]
There is a unique stable match, given by $\{m_1w_2, m_2w_3, m_3w_1\}$.
\end{example}

\begin{lemma}
For the market described in Example \ref{ex:notnecessary}, no coalition blocks the mechanism that outputs a uniform random matching.
\end{lemma}
\begin{proof}
Because the random matching gives each agent their first choice with positive probability, if agent $a$ is in a blocking coalition, then it must be that the agent that $a$ most prefers is also in this coalition. Furthermore, any blocking mechanism must always match all participants, and thus any blocking coalition must have an equal number of men and women. Thus, the only possible blocking coalitions are $\{m_2, m_3, w_1, w_2\}$ or all six agents. The first coalition cannot block; if the probability that $m_2$ and $w_2$ are matched exceeds $1/3$, $m_2$ will not participate. If the probability that $m_3$ and $w_2$ are matched exceeds $1/3$, then $w_2$ will not participate. But at least one of these quantities must be at least $1/2$. 

Considering a mechanism that all agents participate in, for any set of weights on the six possible matchings, we can explicitly write inequalities saying that each agent must get their first choice with probability at least $1/3$, and their last with probability at most $1/3$. Solving these inequalities indicates that any random matching $\mu$ that (weakly) dominates a uniform random matching must satisfy 
\[\Pr(\mu = \{m_1w_1,m_2w_2,m_3w_3\}) = \Pr( \mu = \{m_1w_2,m_2w_3,m_3w_1\}) = \Pr(\mu = \{m_1w_3,m_2w_1,m_3w_2\}),\] 
\[\Pr(\mu = \{m_1w_1,m_2w_3,m_3w_2\}) = \Pr(\mu = \{m_1w_2,m_2w_1,m_3w_3\}) = \Pr(\mu = \{m_1w_3,m_2w_2,m_3w_1\}).\]
But any such mechanism gives each agent their first, second and third choices with equal probability, and thus does not strictly dominate the uniform random matching.
\end{proof}

Finally, the following lemma establishes a simple necessary condition for ex-post incentive compatibility.  This condition will be useful for establishing non-existence of stable outcomes under other notions of stability.

\begin{lemma} \text{ } \\\label{lem:oneone}
If mechanism $\phi$ is ex-post pairwise stable, then if man $m$ and woman $w$ rank each other first under $P$, it follows that $\Pr(\phi(P)(m)=w) = 1$.
\end{lemma}
\begin{proof}
This follows immediately: if $\phi(P)$ matches $m$ and $w$ with probability less than one, then $m$ and $w$ can deviate and match to each other, and both strictly benefit from doing so.
\end{proof}

\subsection{Interim Stability}

The fact that a mechanism which (on fixed input) outputs a uniform random matching is ex-post stable suggests that our notion of a blocking coalition, which relies on ordinal stochastic dominance, is very strict, and that many mechanisms may in fact be stable under this definition even with incomplete information. We show in this section that this intuition is incorrect: despite the strictness of our definition of a blocking coalition, it turns out that \emph{no} mechanism is interim stable.

\begin{theorem} \text{ } \label{thm:interim} \\
No mechanism is interim stable.
\end{theorem}

\begin{proof}
In the proof, we refer to \emph{permutations} of a given preference profile $P$, which informally are preference profiles that are equivalent to $P$ after a relabeling of agents. Formally, given a permutation $\sigma$ on  the set $M \cup W$ which satisfies $\sigma(M) = M$ and $\sigma(W) = W$, we say that $P'$ is the {\bf permutation of $P$ obtained by $\sigma$} if for all $a  \in M \cup W$ and $a'$ in the domain of $P_a$, it holds that $P_a(a') = P'_{\sigma(a)}(\sigma(a'))$.

The proof of Theorem \ref{thm:interim} uses the following example.
\begin{example} \label{ex:interim} 
Suppose that each agent's preferences are iid uniform over the other side, and consider the following preference profile, which we denote $P$:
\[\begin{array}{ l l l l}
P_{m_1} : & w_1 , w_2 , w_3 & \hspace{.6 in} P_{w_1} :& m_1 , m_2 , m_3 \\
P_{m_2}: & w_1 , w_3 , w_2 & \hspace{.6 in} P_{w_2} :& m_1, m_3 , m_2 \\
P_{m_3} : & w_3 , w_1 , w_2 &\hspace{.6 in} P_{w_3} :& m_3 , m_1 , m_2
\end{array}\]
\end{example}
Note that under profile $P$, $m_1$ and $w_1$ rank each other first, as do $m_3$ and $w_3$. By Lemma \ref{lem:nesting}, if $\phi$ is interim stable, it must be ex-post stable. By Lemma \ref{lem:oneone}, given this $P$, any ex-post stable mechanism must produce the match $\{m_1w_1, m_2w_2, m_3w_3\}$ with certainty. Furthermore, if preference profile $P'$ is a permutation of $P$, then 
the matching $\phi(P')$ must simply permute $\{m_1w_1,m_2w_2,m_3w_3\}$ accordingly. Thus, on any permutation of $P$, $\phi$ gives four agents their first choices, and two agents their third choices.

Define the mechanism $\phi'$ as follows: 
\begin{itemize}
\item If $P'$ is the permutation of $P$ obtained by $\sigma$, then \vspace{-.1 in}
\[\phi'(P') = \{\sigma(m_1)\sigma(w_2), \sigma(m_2)\sigma(w_1), \sigma(m_3) \sigma(w_3)\}. \vspace{-.15 in} \]
\item On any profile that is not a permutation of $P$, $\phi'$ mimics $\phi$. 
\end{itemize}
Note that on profile $P$, $\phi'$ gives four agents their first choices, and two agents their second choices. If each agent's preferences are iid uniform over the other side, then
each agent considers his or herself equally likely to play each role in the profile $P$ (by symmetry, this is true even after agents observe their own preferences, as they know nothing about the preferences of others). Thus, conditioned on the preference profile being a permutation of $P$, all agents' interim expected allocation under $\phi$ offers a $2/3$ chance of getting their first choice and a $1/3$ chance of getting  their third choice, while their interim allocation under $\phi'$ offers a $2/3$ chance of getting their first choice and a $1/3$ chance of getting their second choice. Because $\phi'$ and $\phi$ are identical on profiles which are not permutations of $P$, it follows that all agents strictly prefer $\phi'$ to $\phi$ ex-ante.
\end{proof}

The intuition behind the above example is as follows. Stable matchings may be ``inefficient", meaning that it might be possible to separate a stable partnership $(m_1,w_1)$ at little cost to $m_1$ and $w_1$, while providing large gains to their new partners (say $m_2$ and $w_2$). When agents lack the information necessary to determine whether they are likely to play the role of $m_1$ or $m_2$, they will gladly go along with the more efficient (though ex-post unstable) mechanism.

Note that in addition to proving that no mechanism is interim stable \emph{for all priors}, Example \ref{ex:interim} demonstrates that when the priori $\psi$ is (canonically) taken to be uniform on $\mathcal{P}$, there exists no mechanism which is interim stable \emph{at the prior $\psi$} (this follows because if $\phi$ sometimes fails to match pairs who rank each other first, then such pairs have a strict incentive to deviate; if $\phi$ always matches mutual first choices, then all agents prefer to deviate to the mechanism $\phi'$ described above).

Although Theorem \ref{thm:interim} establishes that it is impossible to design a mechanism $\phi$ that eliminates profitable deviations, note that the deviating coalition in Example \ref{ex:interim} involves six agents, and the contract $\phi'$ is fairly complex.
In many settings, such coordinated action may seem implausible, and one might ask whether there exist mechanisms that are at least immune to deviations by \emph{pairs} of agents. The following theorem shows that the complexity of Example \ref{ex:interim} is necessary: any mechanism that always produces a stable match is indeed interim pairwise stable.\footnote{
This result relies crucially on the fact that we're using the notion of stochastic dominance to determine blocking pairs. If agents instead evaluate lotteries over matches by computing expected utilities, it is easy to construct examples where two agents rank each other second, and both prefer matching with certainty to the risk of getting a lower-ranked alternative from $\phi$ (see Appendix \ref{appendix:insurance}).
}

\begin{lemma} \text{ } \label{thm:interimpairwise} \\
Any mechanism which produces a stable match with certainty is interim pairwise stable.
\end{lemma}
\begin{proof}
Seeking a contradiction, suppose that $\phi$ always produces a stable match. Fix a man $m$, and a woman $w$ with whom he might block $\phi$ in the interim.
Note that $m$ must prefer $w$ to going unmatched; otherwise, no deviation with $w$ can strictly benefit him. Thus, the best outcome (for $m$) from a contract with $w$ is that they are matched with certainty. According to the definition of an interim blocking pair, $m$ must believe that receiving $w$ with certainty stochastically dominates the outcome of $\phi$; that is to say, $m$ must be certain that $\phi$ will give him nobody better than $w$. Because $\phi$ produces a stable match, it follows that in cases where $m$ chooses to contract with $w$, $\phi$ always assigns to $w$ a partner that she (weakly) prefers to $m$, and thus she will not participate.
\end{proof}

\subsection{Ex-ante Stability}
In some settings, it is natural to model agents as being uncertain not only about the rankings of others, but also about their own preferences. One might hope that the result of Theorem \ref{thm:interimpairwise} extends to this setting; that is, that if $\phi$ produces a stable match with certainty, it remains immune to pairwise deviations ex-ante. Theorem \ref{thm:exante} states that this is not the case: ex-ante, no mechanism is even pairwise stable.  

\begin{theorem} \label{thm:exante} \text{ } \\
No mechanism is ex-ante pairwise stable.
\end{theorem}

\begin{proof}
The proof of Theorem \ref{thm:exante} uses the following example.
\begin{example} \label{ex:exante}
Suppose that there are three men and three women, and fix $p \in (0,1/4)$. The prior $\psi$ is that preferences are drawn independently as follows: \\

$\begin{array}{l l}
P_{m_1} = \left\{ \begin{array}{l l}  
w_1 , w_3 , w_2  & w.p. \,\, 1-2p \\
 w_2 , w_1 , w_3 & w.p. \,\, p\\
w_3 , w_2 , w_1 & w.p. \,\, p\\
\end{array} \right.
&
P_{w_1} = \left\{ \begin{array}{l l}  
m_1 , m_3 , m_2  & w.p. \,\, 1-2p \\
 m_2 , m_1 , m_3 & w.p. \,\, p\\
m_3 , m_2 , m_1 & w.p. \,\, p\\
\end{array} \right.  \\ 
P_{m_2} =  \hspace{.21 in} w_1 , w_2 & P_{w_2} = \hspace{.21 in} m_1,  m_2\\
P_{m_3} =  \hspace{.21 in} w_3 & P_{w_3} = \hspace{.21 in} m_3
\end{array}$ \\ \\
\end{example}

Because $m_3$ and $w_3$ always rank each other first, we know by Lemmas \ref{lem:nesting} and \ref{lem:oneone} that if mechanism $\phi$ is ex-ante pairwise stable, it matches $m_3$ and $w_3$ with certainty. Applying Lemma \ref{lem:oneone} to the submarket $(\{m_1,m_2\},\{w_1,w_2\})$, we conclude that
\begin{enumerate}
\item Whenever $m_1$ prefers $w_2$ to $w_1$, $\phi$ must match $m_1$ with $w_2$ (and $m_2$ with $w_1$) with certainty. \vspace{-.3 in}
\item Whenever $w_1$ prefers $m_2$ to $m_1$, $\phi$ must match $w_1$ with $m_2$ (and $m_1$ with $w_2$) with certainty. \vspace{-.3 in}
\item Whenever $m_1$ prefers $w_1$ to $w_2$ and $w_1$ prefers $m_1$ to $m_2$, $\phi$ must match $m_1$ with $w_1$. 
\end{enumerate}
After doing the relevant algebra, we see that $w_1$ and $m_1$ each get their first choice with probability $1 - 3p + 4p^2$, their second choice with probability $p$, and their third choice with probability $2p - 4p^2$. If $w_1$ and $m_1$ were to match to each other, they would get their first choice with probability $1 - 2p$, their second with probability $p$, and their third with probability $p$; an outcome that they both prefer. It follows that $\phi$ is not ex-ante pairwise stable, completing the proof.
\end{proof}

The basic intuition for Example \ref{ex:exante} is similar to that of Example \ref{ex:interim}. When $m_1$ ranks $w_1$ first and $w_1$ does not return the favor, it is unstable for them to match and $m_1$ will receive his third choice. In this case, it would (informally) be more ``efficient" (considering only the welfare of $m_1$ and $w_1$) to match $m_1$ with $w_1$; doing so improves the ranking that $m_1$ assigns his partner by two positions, while only lowering the ranking that $w_1$ assigns her partner by one. Because men and women play symmetric roles in the above example, ex-ante, both $m_1$ and $w_1$ prefer the more efficient solution in which they always match to each other.

 \vspace{-.1 in}

\section{Discussion}  \vspace{-.1 in}

In this paper, we extended the notion of stability to settings in which agents are uncertain about their own preferences and/or the preferences of others. We observed that when agents can sign contracts before preferences are fully known, every matching mechanism is susceptible to unraveling. While past work has reached conclusions which sound similar, we argue that our results are stronger in several ways.

First, previous results have assumed that agents are expected utility maximizers, and relied on particular assumptions about the utilities that agents get from each potential partner. Our work uses the stronger notion of stochastic dominance to determine blocking coalitions, and notes that there may exist opportunities for profitable circumvention of a central matching mechanism even when agents are unwilling to sacrifice the chance of a terrific match in order to avoid a poor one.

Second, not only can every mechanism be blocked under \emph{some} prior, 
but also, for some priors, it is impossible to design a mechanism that is interim stable \emph{at that prior}. This striking conclusion is similar to that of \citet{peivandi-vohra_2013}, who find (in a bilateral transferable utility setting) that for some priors over agent types, every potential mechanism of trade can be blocked.

In light of the above findings, one might naturally ask how it is that many centralized clearing houses have managed to persist. 
One possible explanation is that the problematic priors are in some way ``unnatural" and unlikely to arise in practice. We argue that this is not the case: Example \ref{ex:interim} shows that blocking coalitions exist when agent preferences are independent and maximally uncertain, Example \ref{ex:exante} shows that they may exist even when the preferences of most agents are known, and in Appendix \ref{appendix:example} we show that they may exist even when one side has perfectly correlated (i.e. ex-post identical) preferences. 

A more plausible explanation for the persistence of centralized clearing houses is that although mutually profitable early contracting opportunities may exist, agents lack the ability to identify and/or act on them. To take one example, even when profitable early contracting opportunities can be identified, agents may lack the ability to write binding contracts with one another (whereas our work assumes that they possess such commitment power). We leave a more complete discussion of the reasons that stable matching mechanisms might persist in some cases and fail in others to future work.

\newpage

\bibliographystyle{apalike}
\begin{spacing}{0.8}
\bibliography{../UnravelingBibliography}  
\end{spacing}
\appendix

\section{Interim Pairwise (In)Stability} \label{appendix:insurance}
The following example shows that Theorem \ref{thm:interimpairwise} depends on our stochastic dominance notion of a blocking pair; if agents compare lotteries by computing expected utilities, then pairs of agents might benefit from circumventing a mechanism that always produces a stable match.

\begin{example}
There are three agents on each side. Men $m_2$ and $m_3$ are known to rank women in the order $w_1, w_2, w_3$; $m_1$ has this preference with probability $1-p$, and with probability $p$ ranks $w_3$ first. Symmetrically, women $w_2$ and $w_3$ are known to rank men in the order $m_1,  m_2, m_3$; $w_1$ has this preference with probability $1-p$, and with probability $p$ ranks $m_3$ first.
%
\end{example}

 For any realization, there is a unique stable match; note that when $m_1$ ranks $w_3$ first and $w_1$ ranks $m_1$ first and $m_3$ last, this match gives $w_2$ her least-preferred partner, $m_3$. Under a stable matching mechanism, both $m_2$ and $w_2$ get their first choice with probability $p(1-p)$, their second choice with probability $(1-p)^2+p^2$, and their third choice with probability $p(1-p)$. So long as their utility from their second choice is above their average utility from a lottery over their first and third choices, $m_2$ and $w_2$ prefer matching with one another to the outcome of the stable matching.

%

\section{Perfectly Correlated Preferences} \label{appendix:example} 

Theorem \ref{thm:exante} demonstrates that a stable matching mechanism may be blocked ex-ante by a coalition when preferences are drawn independently and uniformly at random. 

The following example considers an opposite extreme extreme, where one side has identical preferences ex-post. It demonstrates that even in this case, it may be possible for a coalition to profitably deviate ex-ante from a mechanism that always selects the unique stable matching. 

In this appendix, we use the language of ``schools" and ``students," and assume that schools all rank students according to a common test.

\begin{example}
Each student has one of four possible preference profiles, drawn independently: \\

\centerline{$\begin{array}{l l r}
A , B , C & w.p. \,\, &(1-\delta)/2 \\
A , C , B & w.p. \,\, & \delta/2 \\
B , A , C & w.p. \,\, & (1-\delta)/2 \\
B , C , A & w.p. \,\, &\delta/2 \\ \\
\end{array}$}

Schools have aligned preferences ex-post. The possibilities are the following: \\

\centerline{$\begin{array}{l l r}
1 , 2 , 3 & w.p. \,\, & (1-\epsilon)/2 \\
1 , 3 , 2 & w.p. \,\, & \epsilon/2 \\
2 , 1 , 3 & w.p. \,\, & (1-\epsilon)/2 \\
2 , 3 , 1 & w.p. \,\, & \epsilon/2 \\ \\
\end{array}$}
%
%

\end{example}
{\onehalfspacing
If all agents participate in an assortative match, schools $A$ and $B$ get their first, second, and third choices with probabilities $(\frac{1}{2}, \frac{1-\delta}{2}, \frac{\delta}{2})$ respectively. Students $1$ and $2$ get their first, second, and third choices with probabilities $\left( \frac{3}{4},\frac{1}{4},0\right) - \frac{\epsilon}{8}\left( 2 - \delta, 2 - 5 \delta + 3 \delta^2, -4 + 6 \delta - 3 \delta^2 \right)$.

If only $(A,B,1,2)$ participate in an assortative match, then the associated match probabilities for schools $A$ and $B$ are $(\frac{1}{2}, \frac{1-\epsilon}{2},\frac{\epsilon}{2})$, and for students $1$ and $2$ are $\left( \frac{3}{4},\frac{1}{4},0\right) - \frac{\delta}{4} \left(0,1,-1 \right)$. 

All four of $A, B, 1, 2$ prefer the latter option if $\epsilon < \delta < 2 \epsilon (1- \frac{3}{2} \delta + \frac{3}{4} \delta^2)$.

}

\end{document}